\documentclass[a4paper,UKenglish,cleveref, autoref]{lipics-v2019}

\nolinenumbers

\usepackage{amsmath,amssymb,array}
\usepackage{amsthm}
\usepackage{verbatim}
\usepackage{float}
\usepackage{tikz}
\usetikzlibrary{arrows,snakes,matrix,backgrounds}

\usepackage{microtype}

\title{Top-$k$ Overlapping Densest Subgraphs: Approximation and Complexity}

\author{Riccardo Dondi}{Universit\`a degli Studi di Bergamo, Bergamo, Italy}{riccardo.dondi@unibg.it}{}{}
\author{Mohammad Mehdi Hosseinzadeh}{Universit\`a degli Studi di Bergamo, Bergamo, Italy}{m.hosseinzadeh@unibg.it}{}{}
\author{Giancarlo Mauri}{Universit\`a degli Studi di Milano-Bicocca, 
Milano, Italy}{mauri@disco.unimib.it}{}{} 
\author{Italo Zoppis}{Universit\`a degli Studi di Milano-Bicocca, 
Milano, Italy}{zoppis@disco.unimib.it}{}{}

\authorrunning{R. Dondi, M. M. Hosseinzadeh, G. Mauri, I. Zoppis} 

\Copyright{R. Dondi, M. M. Hosseinzadeh, G. Mauri, I. Zoppis}



\usepackage[algoruled,linesnumbered,noend]{algorithm2e}

\usepackage{comment}

\newtheorem{problem}{Problem}
\newtheorem{Applemma}{Lemma}
\newtheorem{Apptheorem}{Theorem}
\newtheorem{property}{Property}

\mathchardef\mhyphen="2D

\newcommand{\TOPK}{{\sf Top-k-Overlapping Densest Subgraphs}}

\newcommand{\TOPT}{{\sf Top-3-Overlapping Densest Subgraphs}}

\newcommand{\DENSEK}{\ensuremath{\mathsf{Densest\mhyphen Distinct\mhyphen Subgraph}}}
\newcommand{\DENGR}{\ensuremath{\mathsf{Densest \mhyphen Subgraph}}}

\begin{document}

\maketitle


\begin{abstract}
A central problem in graph mining is finding dense subgraphs, 
with several applications in different fields,
a notable example being identifying communities. 
While a lot of effort 
has been put in the problem of finding a single dense subgraph, 
only recently
the focus has been shifted to the problem of finding a set of densest subgraphs.
An approach introduced to find possible overlapping subgraphs is the 
{\sf Top-k-Overlapping Densest Subgraphs} problem. 
Given an integer $k \geq 1$, the goal of this problem 
is to find a set of $k$ densest subgraphs
that may share some vertices. 
The objective function to be maximized takes into account 
both the density of the subgraphs and the distance between subgraphs 
in the solution.
The {\sf Top-k-Overlapping Densest Subgraphs} problem has been shown to admit
a $\frac{1}{10}$-factor approximation algorithm. 
Furthermore, the computational complexity
of the problem has been left open. 
In this paper, we present contributions concerning the approximability
and the computational complexity of the problem.
For the approximability, we present approximation algorithms that improve the
approximation factor to $\frac{1}{2}$, when $k$ is smaller than the number of vertices in the graph,
and to $\frac{2}{3}$, when $k$ is a constant.
For the computational complexity,
we show that the problem is NP-hard even when $k=3$.
\end{abstract}

\section{Introduction}
\label{sec:Introduction}

One of the most studied and central problems in 
graph mining is the identification
of cohesive subgraphs. This problem has been raised in several contexts,
from social network analysis \cite{DBLP:journals/cn/KumarRRT99} to finding functional motifs in biological networks \cite{bioinfo2006}. 
Different definitions of cohesive graphs have been proposed and applied
in literature. One of the most remarkable example is clique, and finding a maximum size clique is a well-known
and studied problem in theoretical computer science \cite{DBLP:conf/coco/Karp72}. 

Most of the definitions of cohesive subgraph lead to NP-hard problems, 
in some cases even hard to approximate. 
For example, finding a clique of maximum size in a graph $G=(V,E)$
is an NP-hard problem \cite{DBLP:conf/coco/Karp72} and it is even hard to approximate within factor $O(|V|^{1 - \varepsilon})$, for each $\varepsilon > 0$ \cite{DBLP:journals/toc/Zuckerman07}. 
A definition of dense subgraph that leads to a polynomial-time algorithm
is that of average-degree density.  
For this problem, called {\sf Densest Subgraph},
 Goldberg gave
an elegant polynomial-time algorithm \cite{Goldberg:1984:FMD:894477}.
Furthermore, a linear-time greedy algorithm that achieves an approximation 
factor of $\frac{1}{2}$ for {\sf Densest Subgraph} has been given in \cite{DBLP:conf/swat/AsahiroITT96,DBLP:conf/approx/Charikar00}.

The {\sf Densest Subgraph}
 problem aims at finding a single
subgraph, but in many applications it
is of interest finding a collection
of dense subgraphs of a given graph.
More precisely, it is interesting to  compute 
a collection of subgraphs having maximum density in a given graph.
A recent approach proposed in \cite{DBLP:journals/datamine/GalbrunGT16}
asks for a collection of top $k$ densest, possibly
overlapping, subgraphs (denoted as \TOPK{}), 
since in many real-world cases dense subgraphs are related
to non disjoint communities. 
As pointed out in \cite{DBLP:journals/im/LeskovecLDM09,DBLP:journals/datamine/GalbrunGT16},
for example hubs are vertices that may be
part of several communities and hence of several
densest subgraphs, thus motivating the quest for overlapping distinct subgraphs.
\TOPK{}, proposed in \cite{DBLP:journals/datamine/GalbrunGT16},
addresses this problem by 
looking for a collection of $k$ subgraphs that maximize an objective function that takes 
into account both the density of the subgraphs and the distance between the subgraphs
of the solution, thus allowing an overlap
between the subgraphs which depends 
on a parameter $\lambda$.
When $\lambda$ is small, 
then
the density plays a dominant role in the
objective function, so the output subgraphs
can share a significant part of vertices. 
On the other hand, if $\lambda$ is large, 
then the subgraphs will share few or no
vertices, so the subgraphs may be disjoint.


An approach similar to \TOPK{}
was proposed in \cite{DBLP:conf/wsdm/BalalauBCGS15},
where the goal is to find a set of $k$ subgraphs
of maximum density, such that the maximum pairwise
Jaccard coefficient of the subgraphs
is bounded.
A dynamic variant  of the problem, whose goal is finding a set of $k$ disjoint subgraphs,
has been recently considered in \cite{DBLP:conf/cikm/NasirGMG17}.




\TOPK{} has been shown to be approximable within factor $\frac{1}{10}$ \cite{DBLP:journals/datamine/GalbrunGT16},
while its computational complexity has been left open
\cite{DBLP:journals/datamine/GalbrunGT16}.
In this paper, we present algorithmic and complexity results for \TOPK{} when $k$ is less than the number of 
vertices in the graph. 
This last assumption (required in Section \ref{sec:Approximation})
is reasonable, for example notice 
that in the experimental results presented 
in \cite{DBLP:journals/datamine/GalbrunGT16}
$k$ is equal to $20$, even for graphs having thousands or millions of vertices.
Concerning the
approximation of the problem, we provide
in Section \ref{sec:Approximation}
a $\frac{2}{3}$-approximation algorithm when $k$ is a
constant, and we present a
$\frac{1}{2}$-approximation algorithm when $k$ is smaller
than the size of the vertex set. 
From the computational complexity point of view, 
we show in Section \ref{sec:complexity} that 
{\sf Top-k Overlapping Densest Subgraphs}
 is NP-hard even if $k=3$  (that is we ask for three densest subgraphs), 
when $\lambda = 3|V|^3$, for an input
graph $G=(V,E)$.
We conclude the paper in Section \ref{sec:conclusion} with some open problems.
Some of the proofs and the pseudocode of some algorithms are omitted due to page limit.

\section{Definitions}
\label{sec:Def}

In this section, we present some definitions that will be useful in the rest of the paper. 
Moreover, we provide the formal definition of the problem we are interested in.

All the graphs we consider in this paper are undirected.
Given a graph $G=(V,E)$, and a set $V' \subseteq V$, 
we denote by $G[V']$ the \emph{subgraph} of $G$ induced
by $V'$, formally $G[V']=(V',E')$, where $E'$ is defined as follows:
$
E'= \{ \{ u,v \}: \{ u,v \} \in E \wedge u,v \in V'  \}
$.
If $G[V']$ is a subgraph of $G[V'']$, with $V' \subseteq V'' \subseteq V$,
then $G[V'']$ is 
a \emph{supergraph} of $G[V']$.
$G[V']$ is a proper subgraph of $G[V'']$, if $V' \subset V'' \subseteq V$;
in this case $G[V'']$ is a proper supergraph of $G[V']$.
A subgraph $G[V']$ of $G$ is a \emph{singleton}, if $|V'|=1$.

Given a subset $U \subseteq V$, we denote by $E(U)$ the set of edges of
$G$ having both endpoints in $U$. Moreover, we denote
by $E(V_1,V_2)$, with $V_1 \subseteq V'$, $V_2 \subseteq V'$ 
and $V_1 \cap V_2 = \emptyset$, 
the set of edges having exactly one endpoint in $V_1$ and
exactly one endpoint in $V_2$, formally
$
E(V_1,V_2)=\{\{u,v\}: u\in V_1, v \in V_2 \}
$.
Two subgraphs $G[V_1]$ and $G[V_2]$ of a graph $G=(V,E)$ are called \emph{distinct} when $V_1 \neq V_2$.

Next, we present the definition of \emph{crossing subgraphs},
which is fundamental in Section~\ref{subsec:approxNotConstant}.

\begin{definition}
Given a graph $G=(V,E)$, let
$G[V_1]$ and $G[V_2]$ be two subgraphs
of $G=(V,E)$. $G[V_1]$ and $G[V_2]$ are 
\emph{crossing} when $V_1 \cap V_2 \neq \emptyset$, $V_1 \setminus V_2 \neq \emptyset$
and $V_2 \setminus V_1 \neq \emptyset$ (notice that $V_1 \nsubseteq V_2$ and $V_2 \nsubseteq V_1$).
\end{definition}

Consider the example of Fig. \ref{fig:DenseSubgraph}. 
The two subgraphs induced by $\{ v_5,v_6, v_7, v_8, v_9, v_{10} \}$
and $\{ v_1, v_2, v_3, v_4, v_5 \}$ are crossing, while
the two subgraphs induced by $\{ v_5,v_6, v_7, v_8, v_9, v_{10} \}$
and $\{ v_5,v_6, v_7, v_8, v_9 \}$ are not crossing.


Now, we present the definition of density of a subgraph.

\begin{definition}
Given a graph $G=(V,E)$ and a subgraph $G[V']=(V',E')$,
with $V' \subseteq V$, the density of $G[V']$,
denoted by $dens(G[V'])$, is defined as 
$
dens(G[V'])=\frac{|E'|}{|V'|}
$.
\end{definition}

A \emph{densest subgraph} of a graph $G=(V,E)$ is a 
subgraph $G[U]$, with $U \subseteq V$, 
that maximizes $dens(G[U])$, among the subgraphs
of $G$. 
In the example of Fig. \ref{fig:DenseSubgraph}
the subgraph induced by $\{ v_5, v_6, v_7, v_8, v_9, v_{10} \}$ is the densest subgraph
and has density $\frac{11}{6}$.

Given a graph $G=(V,E)$ and 
a collection of subgraphs $\mathcal{W}= \{G[W_1], 
\dots, G[W_k] \}$
where each $G[W_i]$ is a subgraph of $G$, that is $W_i \subseteq V$, with $1 \leq i \leq k$, then the density of $\mathcal{W}$,
denoted by $dens(\mathcal{W})$, 
is defined as 
$
dens(\mathcal{W}) = \sum_{i=1}^{k} dens (G[W_i])
$. 

The goal of the problem we are interested in
is to find a collection of $k$, with $1 \leq k < |V|$, 
possibly overlapping subgraphs having high density. However, 
allowing overlap
leads to a solution that may contain $k$ copies of
the same subgraph.
To address such an issue, in \cite{DBLP:journals/datamine/GalbrunGT16}
a distance function between subgraphs of the collection is included in 
the objective function (to be maximized). We present here the distance
function between two subgraphs presented in 
\cite{DBLP:journals/datamine/GalbrunGT16}.

\begin{definition}
Given a graph $G=(V,E)$ and two subgraphs $G[U]$, $G[Z]$, 
with $U,Z \subseteq V$, define the distance function 
$d: 2^{G[V]} \times 2^{G[V]} \rightarrow \mathbb{R_{+}}$ 
between $G[U]$ and $G[Z]$
as follows:

\begin{equation*}
d(G[U],G[Z]) = \left\{
\begin{array}{ll}
2-\frac{|U \cap Z|^2}{|U||Z|} & \text{if } U \neq Z,\\
0 & \text{else}.
\end{array} \right.
\end{equation*}
\end{definition}

Notice that $d(G[U],G[Z]) \leq 2$, for each 
$U,Z \subseteq V$.

Now, we are able to define the problem we are interested in.

\begin{problem}\TOPK{} \\
\noindent
\textbf{Input:} a graph $G=(V,E)$, a parameter $\lambda > 0$. \\
\textbf{Output:} a set $\mathcal{W} = \{ G[W_1], \dots , G[W_k] \}$
of $k$ subgraphs, with $1 \leq k <|V|$  and $W_i \subseteq V$,
$1 \leq i \leq k$,
that maximizes the following value
\[
r(\mathcal{W})= dens(\mathcal{W})+ \lambda \sum_{i=1}^{k-1} \sum_{j=i+1}^k d(G[W_i],G[W_j]) 
\]
\end{problem}
We assume in what follows that $|V| > 5$ (it is required in the proof of Lemma \ref{lem:secondPartCorrect}). Notice that,
when $|V| \leq 5$, \TOPK{} can be solved optimally
in constant time. 




\begin{figure}
\centering
\includegraphics[scale=.45]{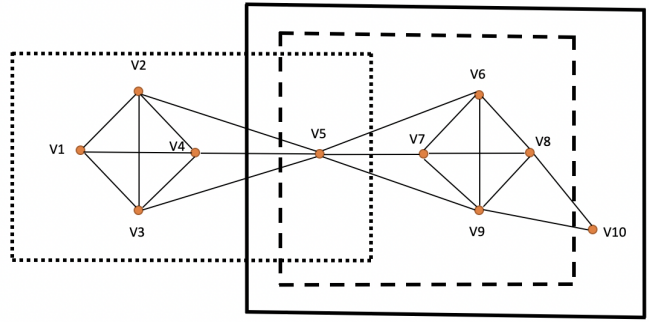}
\caption{A graph and a solution $\mathcal{W}$ of \TOPK{}, for $k=3$, consisting of the three subgraphs
included in boxes.
}
\label{fig:DenseSubgraph}
\end{figure}

\subsection{Goldberg's Algorithm and Extended
Goldberg's Algorithm}

Goldberg's Algorithm \cite{Goldberg:1984:FMD:894477} computes 
in polynomial time an optimal solution for the
\textsf{Densest-Subgraph}
problem that, 
given as input a graph $G=(V,E)$, asks for a 
subgraph $G[V']$ in $G$ having maximum density. 
Goldberg's Algorithm reduces $\DENGR$ to the problem 
of computing a minimum cut in a weighted auxiliary graph 
computed
by adding two vertices $s$ (the sink) and $t$ (the target) 
to $V$, where both $s$ and $t$ are connected
to every vertex of $V$. 
The time complexity of Goldberg's Algorithm is $O(|V||E|\log(|V|^2/|E|))$ applying
parametric flow algorithm \cite{DBLP:journals/siamcomp/GalloGT89}.
Given a graph $G=(V,E)$ and a subgraph $G[V']$,
with $V' \subseteq V$, we denote by 
\textit{Dense-Subgraph}$(G[V'])$ a densest subgraph
in $G[V']$, which can be computed with Goldberg's Algorithm. 

In this paper, we consider also a modification 
of Goldberg's Algorithm given in \cite{Zou2013}.
We refer to this algorithm as the Extended
Goldberg's Algorithm.
Extended Goldberg's Algorithm \cite{Zou2013} 
addresses a constrained variant
of $\DENGR$, where some vertices are forced to be in a densest subgraph,
that is we want to compute a densest subgraph 
$G[V']$ constrained
to the fact that a set $S \subseteq V'$.
We denote
by $Dense\mhyphen Subgraph(G[V'],C(S))$ a densest 
subgraph of $G[V']$ that is forced to contain $S$, 
where 
$S$ is called the \emph{constrained set} of $Dense\mhyphen Subgraph(G[V'],C(S))$.
Notice that 
\textit{Dense-Subgraph}$(G[V'],C(S))$ 
can be computed with the Extended Goldberg's Algorithm 
in time $O(|V||E|\log(|V|^2/|E|))$ \cite{DBLP:journals/siamcomp/GalloGT89,Zou2013}.

\section{Approximating \TOPK{}}
\label{sec:Approximation}

In this section, we present a 
$\frac{2}{3}$-approximation  algorithm for \TOPK{} when $k$ 
is a constant and a $\frac{1}{2}$-approximation algorithm
when $k$ is not a constant.
First, the two approximation algorithms compute a densest subgraph of $G$, 
denoted by $G[W_1]$.
Then, they 
iteratively compute a solution for   
an intermediate problem, called 
{\sf Densest-Distinct-Subgraph}.
When $k$ is constant we are able to solve the $\DENSEK$ problem
in polynomial time, while for 
general $k$ we are able to provide a 
$\frac{1}{2}$-approximation algorithm
for it.

First, we introduce the $\DENSEK$ problem, 
then we present the two
approximation algorithms and the analysis of their approximation factors.


\begin{problem}\DENSEK \\
\noindent
\textbf{Input:} a graph $G=(V,E)$ and a set $\mathcal{W} = \{ G[W_1], \dots , G[W_t] \}$,
with $1 \leq t \leq k-1$,
of subgraphs of $G$.\\ 
\textbf{Output:} a subgraph $G[Z]$ of $G$ such that $Z \neq W_i$, 
for each $1 \leq i \leq t$, and $dens(G[Z])$ is maximum.
\end{problem}

Notice that $\DENSEK$ is not identical to 
compute a densest subgraph of $G$, 
as we need to ensure that the returned subgraph $G[Z]$ is 
distinct from any subgraph in $\mathcal{W}$. 
Moreover, notice that we assume that $|\mathcal{W}| \leq k-1$, since if 
$|\mathcal{W}| = k$ we already have $k$ 
subgraphs in our solution of \TOPK{}.




\subsection{Approximation for Constant $k$}
\label{subsec:ApproxConstant}

First,  we show 
that $\DENSEK$ is polynomial-time 
solvable when $k$ is a constant, 
then we show how to obtain an approximation algorithm
for \TOPK{} 
by iteratively solving $\DENSEK$ and by combining
this solution with one that consists of $k$ singletons.

\subsubsection{A Polynomial-Time Algorithm for $\DENSEK$\\} 

We start by proving a property of solutions of $\DENSEK$. 

\begin{lemma}
\label{lem:DenseK} 
Consider a graph $G=(V,E)$ and a set $\mathcal{W} = \{ G[W_1], \dots , 
G[W_t]\}$, $1 \leq t \leq k-1$,
of subgraphs of $G$. Given a subgraph $G[Z]$ distinct from the subgraphs in
$\mathcal{W}$, there exist at most $t$ 
vertices $u_1, \dots , u_{t}$,
with $u_i \in W_i$, $1 \leq i \leq t$,
that can be partitioned into two sets $U_1$, $U_2$
such that $Z \supseteq U_1$, 
$Z \cap U_2 = \emptyset$ and there is no $G[W_j]$ in $\mathcal{W}$, with $1 \leq j \leq t$,
such that $W_j \supseteq U_1$ and $W_j \cap U_2 = \emptyset$.
\end{lemma}
\begin{proof}
Consider $G[Z]$ and a subgraph $G[W_j]$,
$1 \leq j \leq t$. Since $G[Z]$ is distinct
from $G[W_j]$, it follows that
there exists a vertex $u_j \in V$ such that 
$u_j \in Z \setminus W_j$ (in this case $u_j \in U_1$)
or $u_j \in W_j \setminus Z$ (in this case $u_j \in U_2$),
otherwise $Z$ is identical to $W_j$.
By construction, the two sets $U_1$ and $U_2$ satisfy the lemma.
\end{proof}
Next, based on Lemma \ref{lem:DenseK}, we show how to compute an optimal solution of \textsf{Dense-Distinct-Subgraph}, when $k$ is a constant.
Algorithm \ref{algo:approx3} iterates 
over each subset $\{ u_1,\dots, u_t \}$ of at most $t$ vertices (recall that $|\mathcal{W}|=t$) and 
over the subsets 
$U_1, U_2 \subseteq \{ u_1,\dots, u_t \}$,
such that $U_1 \uplus U_2 = \{ u_1,\dots, u_t \}$. 
Algorithm \ref{algo:approx3} computes a
densest subgraph $G[Z]$ of $G$, 
with constrained set $U_1$ 
and with $Z \cap U_2 = \emptyset$, 
such
that there is no subgraph of $\mathcal{W}$ that contains
$U_1$ and whose set of vertices is disjoint from $U_2$.
Algorithm \ref{algo:approx3} applies the Extended Goldberg's algorithm
on the subgraph $G[ V \setminus U_2]$, with constrained set $C(U_1)$.
We prove the correctness
of Algorithm \ref{algo:approx3} in the next theorem.

\begin{theorem}
\label{teo:OptDenseK}
Let $G[Z]$ be the solution returned by 
Algorithm \ref{algo:approx3}.
Then, an optimal solution $G[Z']$ of $\DENSEK$ over instance
$(G, \mathcal{W})$
has density
at most $dens(G[Z])$.
\end{theorem}

We recall that 
a densest subgraph constrained to a given 
set can be computed in time $O(|V||E| \log(\frac{|V|^2}{|E|}))$ with the
Extended Goldberg's Algorithm \cite{DBLP:journals/siamcomp/GalloGT89,Zou2013}.
It follows that Algorithm~\ref{algo:approx3} returns an optimal solution of $\DENSEK$
in time 
$O(|V|^{2k+2}|E| \log(\frac{|V|^2}{|E|}))$.


\subsubsection{A $\frac{2}{3}$-Approximation Algorithm when $k$ is a Constant\\}

We show that, by solving the $\DENSEK$ problem optimally, 
we achieve a $\frac{2}{3}$ approximation ratio for \TOPK{}. 
The approximation algorithm returns the solution 
of maximum value between the solution returned by 
Algorithm \ref{algo:approxDef} and a
solution consisting of $k$ singletons.

First, we consider the solution
returned by Algorithm \ref{algo:approxDef}.
At each step, Algorithm \ref{algo:approxDef} 
computes an optimal solution of
$\DENSEK$ in time 
$O(|V|^{2k+2}|E| \log(\frac{|V|^2}{|E|}))$
and the output subgraph is added to the solution.
Since $k$ is a constant, hence
the number of iterations of Algorithm \ref{algo:approxDef} is
a constant,
the overall time complexity of Algorithm \ref{algo:approxDef} is $O(|V|^{2k+2}|E| \log(\frac{|V|^2}{|E|}))$. 

\begin{algorithm}[H]
\KwData{a graph $G$ }
\KwResult{ a set $\mathcal{W} = \{ G[W_1], \dots , G[W_k] \}$ of subgraphs of $G$}
$\mathcal{W} \leftarrow \{ G[W_1] \}$ /* $G[W_1]$ is a densest subgraph of $G$ */\;
\For{$i \leftarrow 2$ \KwTo $k$}{
Compute an optimal solution $G[Z]$ of $\DENSEK$ 
with input $(G,\mathcal{W})$ /* Applying Algorithm \ref{algo:approx3} */\;
$\mathcal{W} \leftarrow \mathcal{W} \cup \{ G[Z] \}$
}
Return($\mathcal{W}$)\;
\caption{Algorithm that returns an approximated
solution of \TOPK{}}
\label{algo:approxDef}
\end{algorithm}


Consider the solution
$\mathcal{W}=  \{ G[W_1], \dots, G[W_k]  \}$ returned
by Algorithm \ref{algo:approxDef},
we prove a result on the distance between two
subgraphs of $\mathcal{W}$.

\begin{lemma}
\label{lem:ApproxDistance}
Let $\mathcal{W}=  \{ G[W_1], \dots, G[W_k]  \}$ be a set
of subgraphs returned by Algorithm~\ref{algo:approxDef}.
Then, for each $G[W_i],G[W_j] \in \mathcal{W}$, with $1 \leq i \leq k$, 
$1 \leq j \leq k$ and $i \neq j$, it holds
$d(G[W_i],G[W_j]) > 1$.
\end{lemma}
\begin{proof}
By the definition of distance $d$, 
since $G[W_i]$ and $G[W_j]$, with $1 \leq i \leq k$, 
$1 \leq j \leq k$ and $i \neq j$, 
are distinct subgraphs of $G$, it follows that 
$
d(G[W_i],G[W_j]) = 2- \frac{|W_i \cap W_j|^2}{|W_i||W_j|}
$.
Since $\frac{|W_i \cap W_j|^2}{|W_i||W_j|} \leq 1$, it follows
that 
$d(G[W_i],G[W_j]) \geq 1$.
\end{proof}

Now, we prove a bound on the value $r(\mathcal{W})$
of a solution $\mathcal{W}$
returned by Algorithm \ref{algo:approxDef}.

\begin{lemma}
\label{lem:ApproxRkconstant}
Let $\mathcal{W}=  \{ G[W_1], \dots, G[W_k]  \}$ be a set
of subgraphs returned by Algorithm~\ref{algo:approxDef}
and let $\mathcal{W}^o=  \{ G[W^o_1], \dots, G[W^o_k]  \}$ 
be an optimal solution of \TOPK{} over instance $G$.
Then, for each $\lambda > 0$,
\[
dens(\mathcal{W}) \geq dens(\mathcal{W}^o)
\text{  and  }
\lambda \sum_{i=1}^{k-1} \sum_{j=i+1}^k d(W_i,W_j)
\geq \frac{1}{2} \lambda \sum_{i=1}^{k-1} \sum_{j=i+1}^k d(W_i^o,W_j^o).
\]
\end{lemma}
\begin{proof}
%
The second inequality follows from Lemma \ref{lem:ApproxDistance} and from the fact that
$ d(G[U],G[Z]) \leq 2$, for each $U,Z \subseteq V$.

We prove the first inequality of the lemma 
by induction on $k$.
Let $G[W_i]$, with $2 \leq i \leq k$, be the subgraph added to $\mathcal{W}$ by the $i$-th iteration of
Algorithm \ref{algo:approxDef}.
By construction, $dens(G[W_1])$ $\geq$ $dens(G[W_2]) \geq $ $\dots$ 
$\geq dens(G[W_k])$. Moreover,
assume w.l.o.g. that $dens(G[W^o_1]) \geq$ $dens(G[W^o_2]) \geq $ $\dots$ 
$ \geq dens(G[W^o_k])$.

When $k=1$, by construction of Algorithm \ref{algo:approxDef}, 
$G[W_1]$ is a densest subgraph of $G$,
it follows that $dens(G[W_1]) \geq dens(G[W^o_1])$.
Assume that the lemma holds for $ k -1$, we prove
that it holds for $k$. First, notice that
$
\sum_{i=1}^k dens(G[W_i]) = \sum_{i=1}^{k-1} dens(G[W_i]) +
dens(G[W_k])
$.
By induction hypothesis 
\[
\sum_{i=1}^{k-1} dens(G[W_i]) \geq \sum_{i=1}^{k-1} dens(G[W^o_i]).
\]

Notice that $G[W_k]$ is an optimal solution of $\DENSEK$ on instance
$(G, $ $\{ G[W_1]$, $G[W_2],$ $\dots, G[W_{k-1}] \})$.
By the pigeon principle at least one of 
$G[W^o_1]$, $G[W^o_2]$, $\dots$, $G[W^o_k]$
does not belong to the set $\{ G[W_1], G[W_2],$ $\dots,$ $G[W_{k-1}]\}$ 
of subgraphs, hence, by the optimality of $G[W_k]$,
$dens(G[W_k]) \geq dens(G[W^o_p])$, for some $p$ with $1 \leq p \leq k$, and $dens(G[W^o_p]) \geq dens(G[W^o_k])$.
Now,
\[
\sum_{i=1}^k dens(G[W_i]) = \sum_{i=1}^{k-1} dens(G[W_i]) +
dens(G[W_k])  \geq 
\]
\[
\sum_{i=1}^{k-1} dens(G[W^o_i]) +
dens(G[W^o_k]) \geq \sum_{i=1}^{k} dens(G[W^o_i])
\]
thus concluding the proof.
\end{proof}

Consider Algorithm $A_T$ that, given an instance $G$ of 
\TOPK{},
returns a solution $\mathcal{W}'=\{ G[W'_1], \dots, G[W'_k] \}$ 
consisting of $k$ distinct singletons. 
Notice that, 
since each $G[W'_i]$, with $1 \leq i \leq k$,
is a singleton, it follows that $dens(\mathcal{W}')=0$.
Moreover, since the subgraphs in $\mathcal{W}'$ are pairwise disjoint, we have
$d(G[W'_i],G[W'_j])=2$, 
for each $G[W'_i]$, $G[W'_j] \in \mathcal{W}'$ with $1 \leq i,j \leq k$ 
and $i \neq j$.

We can prove now that the maximum between
$r(\mathcal{W})$ (where $\mathcal{W}$ is the 
solution returned by
Algorithm \ref{algo:approxDef})
and $r(\mathcal{W}')$
(where $\mathcal{W}'$ is the solution returned by Algorithm $A_T$) is at least $\frac{2}{3}$ of the value of an optimal solution of \TOPK{}.

\begin{theorem}
\label{teo:finalApprox}
Let $\mathcal{W}=  \{ G[W_1], \dots, G[W_k]  \}$ 
be the solution 
returned by Algorithm \ref{algo:approxDef} and 
let $\mathcal{W}'=  \{ G[W'_1], \dots, G[W'_k]  \}$
be the solution returned by Algorithm $A_T$.
Let
$\mathcal{W}^o=  \{ G[W^o_1], \dots,$ $G[W^o_k]  \}$ 
be an optimal solution of \TOPK{} over instance $G$. Then
$ 
\max(r(\mathcal{W}), r(\mathcal{W}')) \geq 
\frac{2}{3}~r(\mathcal{W}^o).
$
\end{theorem}

\subsection{Approximation When $k$ is not a Constant}
\label{subsec:approxNotConstant}

Now, we show that \TOPK{} can be approximated
within factor 
$\frac{1}{2}$ when $k$ is not
a constant.
The approximation algorithm (Algorithm \ref{algo:approxkNotConstant}), consists 
of two phases.
In the first phase,
while $\mathcal{W}$ 
does not contain crossing subgraphs
(see Property~\ref{prop:Stop}),
Algorithm~\ref{algo:approxkNotConstant} 
adds to $\mathcal{W}$ a subgraph which is an
optimal solution of $\DENSEK$.
When Property \ref{prop:Stop} holds, Phase 2 of
Algorithm \ref{algo:approxkNotConstant} completes $\mathcal{W}$,
so that $\mathcal{W}$ contains $k$ subgraphs.

\begin{algorithm}
\KwData{a graph $G$ }
\KwResult{$\mathcal{W} = \{ G[W_1], \dots , G[W_k] \}$ of subgraphs of $G$}
$\mathcal{W} \leftarrow \{ G[W_1] \}$ /* $G[W_1]$ is a densest subgraph of $G$ */\;

\textbf{Phase 1}\;
\While{$|\mathcal{W}| < k$ and Property \ref{prop:Stop} does not hold}
{
Compute an optimal solution $G[Z]$ of $\DENSEK$ 
with input $(G,\mathcal{W})$  /* Applying 
Algorithm \ref{algo:approx4} (described later) */\;
$\mathcal{W} \leftarrow \mathcal{W} \cup \{ G[Z] \}$\;
}
\textbf{Phase 2} (Only if $|\mathcal{W}|<k$)\;
$W_{i,j} \leftarrow W_i \cap W_j$, with $W_i$ and $W_j$
two crossing subgraphs in $\mathcal{W}$\;
\If{$|W_{i,j}| \leq 3$}
{
Complete $\mathcal{W}$ by adding the
densest distinct subgraphs (not already in $\mathcal{W}$) induced by 
$W_i \cup \{ v\}$, with $v \in V \setminus W_i$,
and by $W_j \cup \{ u \}$, with $u \in V \setminus W_j$; 
}
\If{$|W_{i,j}| \geq 4$}
{
Complete $\mathcal{W}$ by adding the
densest distinct subgraphs (not already in $\mathcal{W}$) induced by 
$W_i \cup \{ v\}$, with $v \in V \setminus W_i$,
by $W_j \cup \{ u \}$, with $u \in V \setminus W_j$, and by
$W_j \setminus \{ w\}$, with $w \in W_{i,j}$;
}
Return($\mathcal{W}$)\;
%
\caption{Returns an approximated
solution of \TOPK{}}
\label{algo:approxkNotConstant}
\end{algorithm}


First, we define formally the property on which
Algorithm \ref{algo:approxkNotConstant} is based.

\begin{property}
\label{prop:Stop} 
Given a collection $\mathcal{W}$ of $t$ subgraphs, with $2 \leq t \leq k-1$,
there exist two crossing subgraphs $G[W_i]$ and $G[W_j]$ 
in $\mathcal{W}$, with $1 \leq i \leq t$, $1 \leq j \leq t$
and $i \neq j$. 
\end{property}

\subsubsection{Analysis of Phase 1} 

We show that, 
while $\mathcal{W}$ 
does not satisfy Property \ref{prop:Stop},
$\DENSEK$ can be solved optimally in polynomial time.
First, 
we prove a property of a solution of the $\DENSEK$ problem
when Property \ref{prop:Stop} does not hold.

\begin{lemma}
\label{lem:PrelAprrox}
Consider a graph $G=(V,E)$ and a set $\mathcal{W} = \{ G[W_1], \dots , 
G[W_t]\}$, $1 \leq t \leq k-1$,
of subgraphs of $G$ that
does not satisfy Property \ref{prop:Stop}. 
Given a subgraph $G[Z]$ distinct from the subgraphs in
$\mathcal{W}$, there exist at most three vertices $u_1, u_2, u_3 \in V$
such that $ u_1, u_2 \in Z $, $u_3 \notin Z$
and there is no $G[W_j]$ in $\mathcal{W}$, 
$1 \leq j \leq t$, with $u_1, u_2 \in W_j$ and $u_3 \notin W_j$.
\end{lemma}
\begin{proof}
Consider a subgraph $G[Z]$ distinct from the subgraphs in $\mathcal{W}$ and a 
vertex $u_1 \in Z$.
Notice that, for each subgraph in $\mathcal{W}$ that does not contain $u_1$,
the lemma holds.
Now, we consider the set $\mathcal{W}'$ of subgraphs 
in $\mathcal{W}$ that contain $u_1$. 
If $\mathcal{W}' = \emptyset$, then the lemma holds, since 
there is no subgraph in $\mathcal{W}$ that contains $u_1$.

Consider the pair $(\mathcal{W}', \subseteq)$
where $\subseteq$ is the subgraph inclusion relation \footnote{Given $A,B \subseteq V$, $G[A] \subseteq G[B]$ if and only if $A \subseteq B$}.
$(\mathcal{W}', \subseteq)$ is a well-ordered set
\footnote{We recall that a well-ordered set is a pair ($S, \leq$),
where $S$ is a set and $\leq$ is a binary relation on $S$ 
such that (1) Relation $\leq$ satisfies the following properties: reflexivety, antisymmetry, transitivity and comparability; 
(2) every non-empty subset of $S$ has a least element based on
relation $\leq$.}.
Clearly, $\subseteq$ is reflexive, antysimmetric and
transitive on $\mathcal{W}'$. We show that
is comparable, that is, given $G[W_x], G[W_y] \in \mathcal{W'}$ with $W_x \neq W_y$, either
$W_x \subset W_y$ or $W_y \subset W_x$. 
Indeed, consider two subgraphs $G[W_x], G[W_y] \in \mathcal{W}'$, such that
neither $W_x \subset W_y$ nor $W_y \subset W_x$. It follows
that they are crossing subgraphs, since they both contain $u_1$, contradicting the hypothesis
that Property \ref{prop:Stop} does not hold. Since
$\mathcal{W}'$ is a finite set, it follows that
$(\mathcal{W}', \subseteq)$ is a well-ordered set.

Consider now the set $\mathcal{W}'_C$ of subgraphs in $\mathcal{W}'$ that are subgraphs of $G[Z]$ and 
notice that $(\mathcal{W}', \subseteq)$
and $(\mathcal{W}'_C, \subseteq)$
are well-ordered sets.
Let $G[W_v]$ be the largest subgraph in $\mathcal{W}'_C$. 
Since $G[W_v]$ is a subgraph of $G[Z]$, 
there exists a vertex $u_2 \in Z \setminus W_v$.
Since $(\mathcal{W}'_C, \subseteq)$
is a well-ordered set,
each subgraph in $\mathcal{W}'_C \setminus \{ G[W_v] \}$ 
is a subgraph of $G[W_v]$,
thus each subgraph in $\mathcal{W}'_C$ 
does not contain $u_2$.

Consider now the set $\mathcal{W}'_N$ of subgraphs in $\mathcal{W}'$ which are not subgraphs of $G[Z]$. Notice that 
$(\mathcal{W}'_N, \subseteq)$
is a well-ordered set and 
let $G[W_y]$ be the graph of minimum cardinality 
in $\mathcal{W}'_N$.
It follows that there exists a vertex $u_3 \in W_y \setminus Z$, and notice
that, since $(\mathcal{W}'_N, \subseteq)$
is a well-ordered set,
$u_3$ belongs to each subgraph in $\mathcal{W}'_N$.

Since we have shown that there exists a vertex $u_2 \in Z$ that does not belong to 
any subgraph of $\mathcal{W}'_C$ and there exists
a vertex $u_3 \notin Z$ that belongs to each subgraph
of $\mathcal{W}'_N$, the lemma follows.
\end{proof}


Algorithm \ref{algo:approx4} computes an optimal solution $G[Z]$ 
of $\DENSEK$
when Property 1 does not hold.
Algorithm \ref{algo:approx4} is a modified variant of 
Algorithm \ref{algo:approx3} (see Section \ref{subsec:ApproxConstant}), which
considers a set of at most three vertices $u_1,u_2, u_3$,
with $U_1 = \{ u_1,u_2\}$ and $U_2=\{ u_3\}$.
Based on Lemma \ref{lem:PrelAprrox}, we can prove 
the following result.

\begin{theorem}
\label{teo:OptPhase1}
Let $G[Z]$ be the solution returned by Algorithm \ref{algo:approx4}.
Then, an optimal solution of $\DENSEK$ over instance 
$(G,\mathcal{W})$ when
Property \ref{prop:Stop} does not hold has density
at most $dens(G[Z])$.
\end{theorem}

Notice that Algorithm \ref{algo:approx4} returns an optimal solution of $\DENSEK$ in Case 3 
of Lemma \ref{lem:PrelAprrox} in time 
$O(|V|^{4}|E| \log(\frac{|V|^2}{|E|}))$, since it applies the Extended Goldberg's Algorithm 
of complexity $O(|V||E| \log(\frac{|V|^2}{|E|}))$ \cite{Zou2013},
for each set of three vertices in $V$.

\paragraph{\textbf{Analysis of Phase 2}\\ }

Assuming that Property \ref{prop:Stop} holds and 
$|\mathcal{W}|=t < k$, 
we consider Phase 2 of 
Algorithm \ref{algo:approxkNotConstant}.
Given  two crossing subgraphs $G[W_i]$
and $G[W_j]$ of $\mathcal{W}$, with $1 \leq i \leq t$, $1 \leq j \leq t$ and $i \neq j$,
define $W_{i,j} = W_i \cap W_j$,
Algorithm \ref{algo:approxkNotConstant} adds $h = k-t$ subgraphs to $\mathcal{W}$ 
until $|\mathcal{W}|=k$,
as follows.

If $|W_{i,j}| \leq 3$, then Phase 2 of 
Algorithm \ref{algo:approxkNotConstant}
adds the $h$ densest 
distinct subgraphs (not already in $\mathcal{W}$) induced by 
$W_i \cup \{ v\}$, for some $v \in V \setminus W_i$,
and by $W_j \cup \{ u \}$, 
for some $u \in V \setminus W_j$.

If $|W_{i,j}| \geq 4$, then Phase 2 of 
Algorithm \ref{algo:approxkNotConstant} adds
the $h$ densest distinct subgraphs (not already in $\mathcal{W}$)
induced by $W_i \cup \{ v\}$, for some $v \in V \setminus W_i$,
by $W_j \cup \{ u \}$, for some $u \in V \setminus W_j$, 
and by $W_j \setminus \{ w\}$, for some $w \in W_{i,j}$.

Next, we show that, after Phase 2 of 
Algorithm \ref{algo:approxkNotConstant},
$|\mathcal{W}|=k$ and each
subgraph added by Phase 2 has density at least
$\frac{1}{2} dens(G[W_j])$, where $G[W_j]$ is a subgraph added 
to $\mathcal{W}$ in Phase~1.

\begin{lemma}
\label{lem:secondPartCorrect}
$|\mathcal{W}|=k$ after the execution of Phase 2 of 
Algorithm \ref{algo:approxkNotConstant}.
\end{lemma}


\begin{lemma}
\label{lem:secondPart}
Let $G[W']$ be a subgraph added to $\mathcal{W}$ by Phase 2 of 
Algorithm \ref{algo:approxkNotConstant}. Then, 
$dens(G[W']) \geq \frac{1}{2} dens(G[W_j])$, with 
$G[W_j]$ a subgraph added to $\mathcal{W}$ by Phase 1 of 
Algorithm~\ref{algo:approxkNotConstant}.
\end{lemma}

Phase 2 of Algorithm \ref{algo:approxkNotConstant} requires time $O(k^2 |V|)$, 
since we have to compare each subgraph to be added to $\mathcal{W}$
with the subgraphs already in $\mathcal{W}$ and each of this comparison requires
time $O(k|V|)$. Each iteration of Phase 1 
of Algorithm~\ref{algo:approxkNotConstant} requires time 
$O(|V|^{4}|E| \log(\frac{|V|^2}{|E|}))$, hence
the overall complexity of 
Algorithm~\ref{algo:approxkNotConstant} is 
$O(|V|^{5}|E| \log(\frac{|V|^2}{|E|}))$, since Phase 1 is iterated
at most $k \leq |V|-1$ times.


Now, thanks to Lemma \ref{lem:secondPart}, we are able to prove that the density of the solution returned by Algorithm \ref{algo:approxkNotConstant} 
is  at least half the density of an optimal solution
of \TOPK{}.

\begin{lemma}
\label{lem:GeneralApproxLemma}
Let $\mathcal{W}=  \{ G[W_1], \dots, G[W_k]  \}$ be the solution
returned by Algorithm \ref{algo:approxkNotConstant} 
and let
$\mathcal{W}^o=  \{ G[W^o_1], \dots, G[W^o_k]  \}$ be an optimal solution of \TOPK{} over instance $G$.
Then
$\sum_{i=1}^{k} dens(G[W_i]) \geq \frac{1}{2} \sum_{i=1}^{k} dens(G[W^o_i]).
$
\end{lemma}

We can conclude the analysis of the
approximation factor with the following result.

\begin{theorem}
\label{teo:finalApproxGeneralRes}
Let $\mathcal{W}=  \{ G[W_1], \dots, G[W_k]  \}$ be the solution
returned by Algorithm \ref{algo:approxkNotConstant} 
and let
$\mathcal{W}^o=  \{ G[W^o_1], \dots, G[W^o_k]  \}$ be an optimal solution of \TOPK{} over instance $G$. Then
$
r(\mathcal{W}) \geq \frac{1}{2} r(\mathcal{W}^o)
$.
\end{theorem}
\begin{proof}
First, by Lemma \ref{lem:GeneralApproxLemma}, 
it holds $dens(\mathcal{W}) \geq \frac{1}{2} dens(\mathcal{W}^o)$.
Similarly to the proof Lemma \ref{lem:ApproxDistance}, since the subgraphs
in $\{ G[W_1], \dots, G[W_k]  \}$ are all distinct, it holds $d(G[W_i],G[W_j]) \geq 1$, for each $i$, $j$ with $1 \leq i \leq k$, 
$1 \leq j \leq k$ and $i \neq j$,
and by definition $d(G[W^o_i],G[W^o_j]) \leq 2$, thus
$
\lambda \sum_{i=1}^{k-1} \sum_{j=i+1}^k d(G[W_i],G[W_j])
\geq \frac{1}{2} \lambda \sum_{i=1}^{k-1} \sum_{j=i+1}^k d(G[W_i^o],G[W_j^o])
$.
We can conclude that $
r(\mathcal{W}) \geq \frac{1}{2} r(\mathcal{W}^o).$
\end{proof}

\section{Complexity of \TOPK{}}
\label{sec:complexity}

In this section, we consider the computational complexity of \TOPK{} and we show that the problem is NP-hard even if $k=3$. We denote this restriction
of the problem by \TOPT{}.
We prove the result by giving a reduction from 
{\sf 3-Clique Partition}, 
which is NP-complete \cite{DBLP:conf/coco/Karp72}. 
We recall that {\sf 3-Clique Partition}, 
given an input graph $G_P=(V_P,E_P)$, 
asks for a partition of $V_P$ 
into $V_{P,1}$, $V_{P,2}$, $V_{P,3}$
such that $V_P = V_{P,1} \uplus V_{P,2} \uplus V_{P,3}$ and each $G[V_{P,i}]$, 
with $1 \leq i \leq 3$, is a clique.

Given an instance $G_P=(V_P,E_P)$ of 
{\sf 3-Clique Partition},
the input graph $G=(V,E)$ of \TOPT{} is identical to $G_P=(V_P,E_P)$.
Define $\lambda = 3|V|^3$.
In order to define a reduction from {\sf 3-Clique Partition} to \TOPT{},
we show the following results.

\begin{lemma}
\label{lem:3Case1}
Let $G_P=(V_P,E_P)$ be a graph instance of {\sf 3-Clique Partition} and let 
$G=(V,E)$ be the corresponding graph instance of \TOPT{}.
Given three cliques $G_P[V_{P,1}]$, $G_P[V_{P,2}]$, $G_P[V_{P,3}]$ such that $V_{P,1}$, $V_{P,2}$, $V_{P,3}$ partition $V_P$,
we can compute in polynomial time 
a set $\mathcal{W}= \{G[V_1], G[V_2], G[V_3] \}$ 
such that
$r(\mathcal{W}) \geq \frac{|V| -3}{2} +  18|V|^3$.
\end{lemma}

\begin{lemma}
\label{lem:3Case2}
Let $G_P=(V_P,E_P)$ be a graph instance of 
{\sf 3-Clique Partition} and let 
$G=(V,E)$ be the corresponding graph instance of \TOPT{}.
Given  a solution $\mathcal{W}= \{G[V_1], G[V_2], G[V_3] \}$ of \TOPT{} on
instance $G$, 
with $r(\mathcal{W}) \geq \frac{|V| -3}{2}+  18|V|^3$,
we can compute in polynomial time three cliques
$G_P[V_{P,1}]$, $G_P[V_{P,2}]$, $G_P[V_{P,3}]$ of $G_P$ such that $V_{P,1}$, $V_{P,2}$, $V_{P,3}$  partition $V_P$.
\end{lemma}

We can conclude that \TOPT{} is NP-hard.

\begin{theorem}
\TOPT{} is NP-hard.
\end{theorem}
\begin{proof}
From Lemma \ref{lem:3Case1} and Lemma \ref{lem:3Case2}, it follows that we have described a polynomial-time reduction from 
{\sf 3-Clique Partition} to \TOPT{}.
Since {\sf 3-Clique Partition} is NP-complete \cite{DBLP:conf/coco/Karp72}, it follows
that also \TOPT{} is NP-hard.
\end{proof}

\section{Conclusion}
\label{sec:conclusion}

We have shown that 
\TOPK{}  is NP-hard when $k=3$ and we have given two approximation
algorithms of factor $\frac{2}{3}$ and $\frac{1}{2}$, when $k$ is a constant
and when $k$ is smaller than the number of vertices in the graph, respectively.
For future works, it would be interesting to further investigate
the approximability of \TOPK{}, possibly improving
the approximation factor or improving the time complexity of our approximation algorithms. 
A second interesting open problem
is the computational complexity of \TOPK{}, in particular when $\lambda$ is a constant. 
Another open problem of theoretical interest is the computational 
complexity of \TOPK{} when $k=2$.

\bibliographystyle{splncs03}
\bibliography{biblio1}

\begin{thebibliography}{10}
\providecommand{\url}[1]{\texttt{#1}}
\providecommand{\urlprefix}{URL }

\bibitem{DBLP:conf/swat/AsahiroITT96}
Asahiro, Y., Iwama, K., Tamaki, H., Tokuyama, T.: Greedily finding a dense
  subgraph. In: Karlsson, R.G., Lingas, A. (eds.) Algorithm Theory - {SWAT}
  '96, 5th Scandinavian Workshop on Algorithm Theory, Reykjav{\'{\i}}k,
  Iceland, July 3-5, 1996, Proceedings. Lecture Notes in Computer Science, vol.
  1097, pp. 136--148. Springer (1996)

\bibitem{DBLP:conf/wsdm/BalalauBCGS15}
Balalau, O.D., Bonchi, F., Chan, T.H., Gullo, F., Sozio, M.: Finding subgraphs
  with maximum total density and limited overlap. In: Cheng, X., Li, H.,
  Gabrilovich, E., Tang, J. (eds.) Proceedings of the Eighth {ACM}
  International Conference on Web Search and Data Mining, {WSDM} 2015. pp.
  379--388. {ACM} (2015)

\bibitem{DBLP:conf/approx/Charikar00}
Charikar, M.: Greedy approximation algorithms for finding dense components in a
  graph. In: Jansen, K., Khuller, S. (eds.) Approximation Algorithms for
  Combinatorial Optimization, Third International Workshop, {APPROX} 2000,
  Proceedings. Lecture Notes in Computer Science, vol. 1913, pp. 84--95.
  Springer (2000)

\bibitem{bioinfo2006}
Fratkin, E., Naughton, B.T., Brutlag, D.L., Batzoglou, S.: Motifcut: regulatory
  motifs finding with maximum density subgraphs. Bioinformatics  22(14),
  156--157 (2006)

\bibitem{DBLP:journals/datamine/GalbrunGT16}
Galbrun, E., Gionis, A., Tatti, N.: Top-k overlapping densest subgraphs. Data
  Min. Knowl. Discov.  30(5),  1134--1165 (2016)

\bibitem{DBLP:journals/siamcomp/GalloGT89}
Gallo, G., Grigoriadis, M.D., Tarjan, R.E.: A fast parametric maximum flow
  algorithm and applications. {SIAM} Journal on Computing  18(1),  30--55
  (1989)

\bibitem{Goldberg:1984:FMD:894477}
Goldberg, A.V.: Finding a maximum density subgraph. Tech. rep., Berkeley, CA,
  USA (1984)

\bibitem{DBLP:conf/coco/Karp72}
Karp, R.M.: Reducibility among combinatorial problems. In: Miller, R.E.,
  Thatcher, J.W. (eds.) Proceedings of a symposium on the Complexity of
  Computer Computations. pp. 85--103. The {IBM} Research Symposia Series,
  Plenum Press, New York (1972)

\bibitem{DBLP:journals/cn/KumarRRT99}
Kumar, R., Raghavan, P., Rajagopalan, S., Tomkins, A.: Trawling the web for
  emerging cyber-communities. Computer Networks  31(11-16),  1481--1493 (1999)

\bibitem{DBLP:journals/im/LeskovecLDM09}
Leskovec, J., Lang, K.J., Dasgupta, A., Mahoney, M.W.: Community structure in
  large networks: Natural cluster sizes and the absence of large well-defined
  clusters. Internet Mathematics  6(1),  29--123 (2009)

\bibitem{DBLP:conf/cikm/NasirGMG17}
Nasir, M.A.U., Gionis, A., Morales, G.D.F., Girdzijauskas, S.: Fully dynamic
  algorithm for top-\emph{k} densest subgraphs. In: Lim, E., Winslett, M.,
  Sanderson, M., Fu, A.W., Sun, J., Culpepper, J.S., Lo, E., Ho, J.C., Donato,
  D., Agrawal, R., Zheng, Y., Castillo, C., Sun, A., Tseng, V.S., Li, C. (eds.)
  Proceedings of the 2017 {ACM} on Conference on Information and Knowledge
  Management, {CIKM} 2017. pp. 1817--1826. {ACM} (2017)

\bibitem{Zou2013}
Zou, Z.: Polynomial-time algorithm for finding densest subgraphs in uncertain
  graphs. In: Proceedings of Internation Workshop on Mining and Learning with
  Graphs (2013)

\bibitem{DBLP:journals/toc/Zuckerman07}
Zuckerman, D.: Linear degree extractors and the inapproximability of max clique
  and chromatic number. Theory of Computing  3(1),  103--128 (2007)

\end{thebibliography}

\newpage

\section*{Appendix}

\subsection*{Pseudocode of Algorithm \ref{algo:approx3}}

\begin{algorithm}[H]
\KwData{a graph $G$ and a set $\mathcal{W} = \{ G[W_1], \dots , 
G[W_t] \}$ of subgraphs of $G$}
\KwResult{
a subgraph $G[Z]$ of $G$, with $Z \neq W_i$, 
for each $1 \leq i \leq t$, and $dens(Z)$ is maximum
}
$Z = \emptyset$\;
$dens=0$\;

\For{$U_1$, $U_2$ in $V$, with
$U_1 \cap U_2 = \emptyset$, 
$|U_1 \cup U_2| \leq t$,
such that there is no 
subgraph $G[W_i]$ in 
$\mathcal{W}$ with  
$W_i \supseteq U_1$ and 
$W_i \cap U_2 = \emptyset$
}
{
$G[X] \leftarrow$ Dense-subgraph($G[V \setminus U_2],C(U_1))$\;
$dens' \leftarrow dens(G[X])$\;
\If{$dens' > dens$}{
$dens \leftarrow dens'$\;
$Z \leftarrow X$\;
} 
}
Return($G[Z]$)\;
\caption{Returns an optimal solution for $\DENSEK$ when $k$ is a constant}
\label{algo:approx3}
\end{algorithm}

\subsection*{Proof of Theorem \ref{teo:OptDenseK}}

\setcounter{Apptheorem}{4}

\begin{Apptheorem}
\label{appendix:teo:OptDenseK}
Let $G[Z]$ be the solution returned by 
Algorithm \ref{algo:approx3}.
Then, an optimal solution $G[Z']$ of $\DENSEK$ over instance
$(G, \mathcal{W})$
has density
at most $dens(G[Z])$.
\end{Apptheorem}
\begin{proof}
Consider a collection $\mathcal{W} = \{ G[W_1], \dots , 
G[W_t] \}$ of subgraphs of $G$ and let $G[Z]$ be the solution
returned by Algorithm \ref{algo:approx3}.
By Lemma \ref{lem:DenseK} it follows
that for each subgraph  
distinct from those in $\mathcal{W}$, hence
also for an optimal solution $G[X]$ of  $\DENSEK$ over instance
$(G,\mathcal{W})$,
there exist at most $t$ vertices $u_1, \dots , u_{t}$,
that can be partitioned into two sets $U_1$, $U_2$
such that $X \supseteq U_1$, 
$X \cap U_2 = \emptyset$ and there is no $G[W_j]$ in $\mathcal{W}$, with $1 \leq j \leq t$,
such that $W_j \supseteq U_1$ and $W_j \cap U_2 = \emptyset$.
The subgraph $G[Z]$ returned by Algorithm \ref{algo:approx3} 
is computed as a densest subgraph 
over each subset $U$ of at most
$t$ vertices and for each partition 
of $U$ into two sets $U'_1$ and $U'_2$, 
such that $Z \supseteq U'_1$, 
$Z \cap U'_2 = \emptyset$  
and there is no $G[W_j]$ in $\mathcal{W}$, 
with $1 \leq j \leq t$,
such that $W_j \supseteq U'_1$ and $W_j \cap U'_2 = \emptyset$.
This holds also when $U'_1 =U_1$
and $U'_2 = U_2$, hence $dens(G[Z]) \geq dens(G[X])$.
\end{proof}

\subsection*{Proof of Theorem \ref{teo:finalApprox}}

\setcounter{Apptheorem}{7}

\begin{Apptheorem}
\label{appendix:teo:finalApprox}
Let $\mathcal{W}=  \{ G[W_1], \dots, G[W_k]  \}$ 
be the solution 
returned by Algorithm \ref{algo:approxDef} and 
let $\mathcal{W}'=  \{ G[W'_1], \dots, G[W'_k]  \}$
be the solution returned by Algorithm $A_T$.
Let
$\mathcal{W}^o=  \{ G[W^o_1], \dots,$ $G[W^o_k]  \}$ 
be an optimal solution of \TOPK{} over instance $G$. Then
$ 
\max(r(\mathcal{W}), r(\mathcal{W}')) \geq 
\frac{2}{3}~r(\mathcal{W}^o).
$
\end{Apptheorem}
\begin{proof}
By Lemma \ref{lem:ApproxRkconstant}, it holds 
$dens(\mathcal{W}) \geq dens(\mathcal{W}^o)$ 
and 
\[
\lambda \sum_{i=1}^{k-1} \sum_{j=i+1}^k d(G[W_i],G[W_j])
\geq \frac{1}{2} \lambda \sum_{i=1}^{k-1} \sum_{j=i+1}^k d(G[W_i^o],G[W_j^o]).
\]

Algorithm $A_T$ returns solution 
$\mathcal{W}'=  \{ G[W'_1], \dots, G[W'_k]  \}$ such that

\[
\lambda \sum_{i=1}^{k-1} \sum_{j=i+1}^k d(G[W'_i],G[W'_j])
\geq \lambda \sum_{i=1}^{k-1} \sum_{j=i+1}^k d(G[W_i^o],G[W_j^o]).
\]

Assume that $\lambda \sum_{i=1}^{k-1} \sum_{j=i+1}^k d(G[W_i^o],G[W_j^o]) 
\geq 2~dens(\mathcal{W}^o)$.
Then 
\[
\frac{1}{3}~\lambda \sum_{i=1}^{k-1} \sum_{j=i+1}^k d(G[W'_i],G[W'_j])
\geq \frac{2}{3}~dens(\mathcal{W}^o)
\]
thus, 
\[
\lambda \sum_{i=1}^{k-1} \sum_{j=i+1}^k d(G[W'_i],G[W'_j])
\geq \]
\[
\frac{2}{3}~\lambda \sum_{i=1}^{k-1} \sum_{j=i+1}^k d(G[W_i^o],G[W_j^o]) + 
\frac{1}{3} \lambda \sum_{i=1}^{k-1} \sum_{j=i+1}^k d(G[W'_i],G[W'_j])
\geq
\]
\[
\frac{2}{3}~\lambda \sum_{i=1}^{k-1} \sum_{j=i+1}^k d(G[W_i^o],G[W_j^o]) + \frac{2}{3}~dens(\mathcal{W}^o)
\]
thus in this case $A_T$ returns a solution having approximation
factor $\frac{2}{3}$.

Assume that $\lambda \sum_{i=1}^{k-1} \sum_{j=i+1}^k d(W_i^o,W_j^o) 
< 2~dens(\mathcal{W}^o)$.
It holds 
\[
dens(\mathcal{W}) \geq dens(\mathcal{W}^o)
=
\frac{2}{3}~dens(\mathcal{W}^o) + \frac{1}{3}~dens(\mathcal{W}^o)
> 
\]
\[
\frac{2}{3}~dens(\mathcal{W}^o)  + ~\frac{1}{6} \lambda \sum_{i=1}^{k-1} \sum_{j=i+1}^k d(G[W_i^o],G[W_j^o]). 
\]
By Lemma \ref{lem:ApproxRkconstant}
\[
\lambda \sum_{i=1}^{k-1} \sum_{j=i+1}^k d(G[W_i],G[W_j])
\geq \frac{1}{2} \lambda \sum_{i=1}^{k-1} \sum_{j=i+1}^k d(G[W_i^o],G[W_j^o]).
\]
We can conclude that 
\[
r(\mathcal{W}) > \frac{2}{3}~dens(\mathcal{W}^o)  + 
\frac{1}{2} \lambda \sum_{i=1}^{k-1} \sum_{j=i+1}^k d(G[W_i^o],G[W_j^o]) 
+\frac{1}{6} \lambda \sum_{i=1}^{k-1} \sum_{j=i+1}^k d(G[W_i^o],G[W_j^o]) 
\]
hence 
$r(\mathcal{W}) \geq \frac{2}{3}~r(\mathcal{W}^o)$.
\end{proof}

\subsection*{Pseudocode of Algorithm \ref{algo:approx4}}

\begin{algorithm}[H]
\KwData{a graph $G$ and a set $\mathcal{W} = \{ G[W_1], \dots , 
G[W_t] \}$ of subgraphs of $G$, such that Property \ref{prop:Stop} does not hold}
\KwResult{a subgraph $G[Z]$ of $G$, with $Z \neq W_i$, 
for each $1 \leq i \leq t$, and $dens(Z)$ is maximum}
$Z \leftarrow \emptyset$\;
$dens \leftarrow 0$\;

\For{ Each subset $\{u_1,u_2,u_3\} \subseteq V$ of at most
three vertices, such that there is no 
subgraph $G[W_i]$ in $\mathcal{W}$ with  
$u_1,u_2 \in W_i$ and 
$u_3 \notin W_i$ 
}
{
$G[X] \leftarrow$ Dense-subgraph($G[V \setminus \{ u_3 \}],C(\{u_1,u_2\}))$\;
$dens' \leftarrow dens(G[X])$\;
\If{$dens' > dens$}{
$dens \leftarrow dens'$\;
$Z \leftarrow X$\;
} 
}
Return $(G[Z])$\;
\caption{Returns an optimal solution for $\DENSEK$ when Property \ref{prop:Stop} does not hold}
\label{algo:approx4}
\end{algorithm}

\subsection*{Proof of Theorem \ref{teo:OptPhase1}}

\setcounter{Apptheorem}{9}

\begin{Apptheorem}
\label{appendix:teo:OptPhase1}
Let $G[Z]$ be the solution returned by Algorithm \ref{algo:approx4}.
Then, an optimal solution of $\DENSEK$ over instance 
$(G,\mathcal{W})$ when
Property \ref{prop:Stop} does not hold has density
at most $dens(G[Z])$.
\end{Apptheorem}
\begin{proof}
Given $(G,\mathcal{W})$, consider a subgraph
$G[X]$ of maximal density distinct from the subgraphs in $\mathcal{W}$.
By Lemma \ref{lem:PrelAprrox}, it follows
that there exist at most three vertices $u_1$, $u_2$, $u_3$ such that $u_1, u_2 \in X$ and $u_3 \notin X$
and there is  no subgraph in $\mathcal{W}$  
satisfying the same property.
The subgraph $G[Z]$ returned by
Algorithm \ref{algo:approx4}
is computed as a densest subgraph over each
subset of at most three vertices $u'_1, u'_2, u'_3 \in V$
such that $u'_1, u'_2 \in Z$ and $u'_3 \notin Z$
and there is no subgraph in $\mathcal{W}$  
satisfying the same property, hence
also in the case $u'_i = u_i$, with $1 \leq i \leq 3$.
It follows that  $dens(G[Z])\geq dens(G[X])$.
\end{proof}

\setcounter{Applemma}{10}

\subsection*{Proof of Lemma \ref{lem:secondPartCorrect}}

\begin{Applemma}
\label{appendix:lem:secondPartCorrect}
$|\mathcal{W}|=k$ after the execution of Phase 2 of 
Algorithm \ref{algo:approxkNotConstant}.
\end{Applemma}
\begin{proof}
Recall that we have assumed $|V| > 5$ and 
that $G[W_i]$ and $G[W_j]$
are two crossing subgraphs added in Phase 1 of 
Algorithm \ref{algo:approxkNotConstant}, with
$W_{i,j}= W_i \cap W_j$.

Consider the case that $|W_{i,j}| \leq 3$.
If $|W_i \setminus W_{i,j}| \geq 2$ or $|W_j \setminus W_{i,j}| \geq 2$, 
then 
for each $u \notin W_i \cap W_j$, there exist
at least $|V|-3$ distinct subgraphs induced by $W_i \cup \{ u\}$, with 
$u \notin W_i$, or by $W_j \cup \{ u\}$, with 
$u \notin W_j$.
Since $G[W_i]$, $G[W_j]$ are in $\mathcal{W}$ and $k \leq |V|-1$,
it follows that in this case at least $k$ subgraphs belong to $\mathcal{W}$
after Phase 2 of Algorithm \ref{algo:approxkNotConstant}.

If both $|W_i \setminus W_{i,j}| = 1$ and $|W_j \setminus W_{i,j} | = 1$, 
then for each $u \notin W_i \cap W_j$, there exists
one subgraph induced by $W_i \cup W_j$,
at least $|V|-5$ distinct subgraphs induced by $W_i \cup \{ u\}$,
with $u \in V \setminus (W_i \cup W_j)$,
and at least $|V|-5$ distinct subgraphs induced by $W_j \cup \{ v\}$,
with $v \in V \setminus (W_i \cup W_j)$. Since $|V|> 5 $, it follows
that at least $|V|-5 + |V|-5 + 1 \geq |V|-5 + 2 \geq |V|-3$ distinct subgraphs are induced by 
$W_i \cup \{ u\}$, with 
$u \notin W_i$, or by $W_j \cup \{ u\}$, with 
$u \notin W_j$.
Since $G[W_i]$, $G[W_j]$ are in $\mathcal{W}$ and $k \leq |V|-1$,
it follows that in this case at least $k$ subgraphs belong to $\mathcal{W}$
after Phase 2 of Algorithm \ref{algo:approxkNotConstant}.


Consider now the case that $|W_{i,j}| \geq  4$.
There exist at least $|V \setminus W_i|$ subgraphs induced by 
$W_i \cup \{ v \}$, with $v \in V \setminus W_i$,
and at least $|V \setminus W_j|$ subgraphs induced by 
$W_j \cup \{ u \}$, with $u \in V \setminus W_j$.
Hence there exist at least $|V \setminus W_{i,j}|-1$ subgraphs induced by 
$W_i \cup \{ v \}$, with $v \in V \setminus W_i$, 
or by $W_j \cup \{ u \}$, with $u \in V \setminus W_j$
(notice that $W_i \cup \{ v \}$ and $W_j \cup \{ u \}$
induce identical subgraphs when $W_i = \{u \}$ and 
$W_j = \{ v \}$).

There exist at least $|W_{i,j}|$ subgraphs induced by 
$W_j \setminus \{w\}$, for some $w \in W_i \cap W_j$.
Since $k \leq |V|-1$, it follows that in this case at least 
$k$ subgraphs belong to $\mathcal{W}$
after Phase 2 of Algorithm~\ref{algo:approxkNotConstant}.
%
%
\end{proof}

\setcounter{Applemma}{11}

\subsection*{Proof of Lemma \ref{lem:secondPart}}

\begin{Applemma}
\label{appendix:lem:secondPart}
Let $G[W']$ be a subgraph added to $\mathcal{W}$ by Phase 2 of 
Algorithm \ref{algo:approxkNotConstant}. Then, 
$dens(G[W']) \geq \frac{1}{2} dens(G[W_j])$, with 
$G[W_j]$ a subgraph added to $\mathcal{W}$ by Phase 1 of Algorithm \ref{algo:approxkNotConstant}.
\end{Applemma}
\begin{proof}
Consider $G[W_j]$ and $G[W_i]$, two subgraphs 
added to $\mathcal{W}$ by Phase 1 of Algorithm \ref{algo:approxkNotConstant}, and 
$W_{i,j} = W_i \cap W_j$.
Consider the case that $|W_{i,j}| \leq 3$.
The density of a subgraph induced by 
$W' = W_j \cup \{u\}$,
added by Phase 2 of 
Algorithm \ref{algo:approxkNotConstant} can
be bounded as follows:
\[
dens(G[W']) \geq \frac{E(W_j)}{|W_j|+1}
\geq \frac{E(W_j)}{|W_j|} \frac{|W_j|}{|W_j|+1} = dens(W_j) \frac{|W_j|}{|W_j|+1} \geq \frac{1}{2} dens(G[W_j])
\]
as $|W_j| \geq 1$.

Similarly, if $W' = W_i \cup \{u\}$ then 
\[
dens(G[W']) \geq \frac{1}{2} dens(G[W_i]).
\]


Now, consider the case that 
$|W_{i,j}| \geq  4$.
For a subgraph  added by
Phase 2 of Algorithm \ref{algo:approxkNotConstant} to
$\mathcal{W}$ and induced by 
either $W_j \cup \{u\}$ or $W_i \cup \{v\}$, 
it holds the same argument of the case 
$|W_{i,j}| \leq  3$,
thus, it holds 
\[
dens(G[W']) \geq \frac{1}{2} dens(G[W_j]).
\]

Now, we consider the density of a subgraph
$G[W']$, with $W' = W_j \setminus \{u\}$, 
where $u \in W_{i,j}$, added to $\mathcal{W}$
by Phase 2 of Algorithm \ref{algo:approxkNotConstant}.
Consider the sum of the density of the
subgraphs $G[W']$ over the vertices 
$u \in W_{i,j}$:
\begin{small}
\[
\sum_{u \in W_{i,j}} dens(G[W']) = 
\sum_{u \in W_{i,j}} \frac{1}{|W_j|-1} 
\left(E(W_j \setminus W_{i,j}) + E(W_{i,j} \setminus \{u\})
+ E(W_j \setminus W_{i,j}, W_{i,j} \setminus \{u\})\right).
\]
\end{small}

Each edge $\{v,w\}$, with $v,w \in W_{i,j}$,
is not considered in the sum
\[
\sum_{u \in W_{i,j}} \frac{1}{|W_j|-1} E(W_{i,j} \setminus \{u\})
\]
at most twice, once for $u=v$ and once for $u=w$.
It follows that
\[
\sum_{u \in W_{i,j}} E(W_{i,j} \setminus \{u\}) \geq 
(|W_{i,j}|-2) E(W_{i,j}).
\]

Each edge $\{w,v\}$, with $v \in W_{i,j}$
and $w \in W_j \setminus W_{i,j}$,
is not considered in  
\[
\sum_{u \in W_{i,j}} \frac{1}{|W_j|-1} E(W_j \setminus W_{i,j}, W_{i,j} 
\setminus \{u\})
\]
at most once, when $u=v$, thus

\[
\sum_{u \in W_{i,j}} E(W_j \setminus W_{i,j}, W_{i,j} \setminus \{u\}) \geq 
(|W_{i,j}|-1) E(W_j \setminus W_{i,j}, W_{i,j} ).
\]

Thus
\begin{small}
\[
\sum_{u \in W_{i,j}} dens(G[W']) =  
\sum_{u \in W_{i,j}} \frac{1}{|W_j|-1} 
\left(E(W_j \setminus W_{i,j}) + E(W_{i,j} \setminus \{u\})
+ E(W_j \setminus W_{i,j}, W_{i,j} \setminus \{u\})\right) \geq
\]
\end{small}

\[ 
\frac{1}{|W_j|-1} 
(|W_{i,j}|
E(W_j \setminus W_{i,j}) + 
(|W_{i,j}|-2) E(W_{i,j})
+ (|W_{i,j}|-1)
E(W_j \setminus W_{i,j}, W_{i,j})) \geq
\]

\[
\frac{|W_{i,j}|-2}{|W_j|-1} 
(E(W_j \setminus W_{i,j}) +  E(W_{i,j})
+ E(W_j \setminus W_{i,j}, W_{i,j})). 
\]

Thus
\[
\sum_{u \in W_{i,j}} dens(G[W']) \geq
\frac{|W_{i,j}|-2}{|W_j|-1} (
E(W_j \setminus W_{i,j}) +  E(W_{i,j})
+ E(W_j \setminus W_{i,j}, W_{i,j})) \geq 
\]
\[
(|W_{i,j}|-2) (dens(G[W_j]))  
\]
as 
\[
dens(G[W_j]) = \frac{1}{|W_j|} 
\left(
E(W_j \setminus W_{i,j}) +  E(W_{i,j})
+ E(W_j \setminus W_{i,j}, W_{i,j}) \right) \leq
\]

\[
\frac{1}{|W_j|-1} 
\left(
E(W_j \setminus W_{i,j}) +  E(W_{i,j})
+ E(W_j \setminus W_{i,j}, W_{i,j}) \right).
\]


It follows that
\[
\sum_{u \in W_{i,j}} dens(G[W']) \geq
(|W_{i,j}|-2) (dens(G[W_j]))  =
\frac{(|W_{i,j}|-2)}{|W_{i,j}|} |W_{i,j}| (dens(G[W_j])).
\]

Since $|W_{i,j}|\geq 4$, it follows
that $\frac{(|W_{i,j}|-2)}{(|W_{i,j}|} \geq \frac{1}{2}$,
thus 
\[
\sum_{u \in W_{i,j}} dens(G[W']) \geq 
\frac{1}{2} \sum_{x \in W_{i,j}} dens(G[W_j])
\]
as
$
\sum_{u\in W_{i,j}} dens(G[W_j]) = 
|W_{i,j}| dens(G[W_j])$.

Since, Algorithm \ref{algo:approxkNotConstant} adds the 
$h$ most dense subgraphs 
among
the choice of $u \in W_{i,j}$ so that $|\mathcal{W}|=k$, 
this completes the proof.
\end{proof}

\setcounter{Applemma}{12}

\subsection*{Proof of Lemma \ref{lem:GeneralApproxLemma}}

\begin{Applemma}
\label{appendix:lem:GeneralApproxLemma}
Let $\mathcal{W}=  \{ G[W_1], \dots, G[W_k]  \}$ be the solution
returned by Algorithm \ref{algo:approxkNotConstant} 
and let
$\mathcal{W}^o=  \{ G[W^o_1], \dots, G[W^o_k]  \}$ be an optimal solution of \TOPK{} over instance $G$.
Then
$\sum_{i=1}^{k} dens(G[W_i]) \geq \frac{1}{2} \sum_{i=1}^{k} dens(G[W^o_i]).
$
\end{Applemma}
\begin{proof}

We prove the lemma by induction on $k$.
When $k=1$, since $G[W_1]$ is a densest subgraph of $G$,
it follows that $dens(G[W_1]) \geq \frac{1}{2} dens(G[W^o_1])$.

Assume that the lemma holds for $ k -1$, we prove
that it holds for $k$. First, notice that
\[
\sum_{i=1}^k dens(G[W_i]) = \sum_{i=1}^{k-1} dens(G[W_i]) +
dens(G[W_k]).
\]
By induction hypothesis 
\[
\sum_{i=1}^{k-1} dens(G[W_i]) \geq \frac{1}{2} \sum_{i=1}^{k-1} dens(G[W^o_i]).
\]
We prove that there exists 
a subgraph $G[W_j]$ added to $\mathcal{W}$ by Phase 1 of Algorithm \ref{algo:approxkNotConstant} such that
$dens(G[W_k]) \geq \frac{1}{2} dens(G[W_j])$. This
property clearly holds if $G[W_k]$ is added to $\mathcal{W}$ by Phase 1 of Algorithm \ref{algo:approxkNotConstant}.
If $G[W_k]$ is added to $\mathcal{W}$ by Phase 2 
of Algorithm \ref{algo:approxkNotConstant},
by Lemma \ref{lem:secondPart} it follows that
there exists a subgraph $G[W_j]$ added to $\mathcal{W}$ by Phase 1 of Algorithm \ref{algo:approxkNotConstant} such that
$dens(G[W_k]) \geq \frac{1}{2} dens(G[W_j])$.

Consider an optimal solution $G[W'_k]$ of $\DENSEK$ over instance
$(G, $  $\{ G[W_1]$, $G[W_2],$ $\dots, G[W_{k-1}] \})$.
Since $G[W_j]$ is added to 
$\mathcal{W}$ by Phase 1 of Algorithm \ref{algo:approxkNotConstant},
by Theorem \ref{teo:OptPhase1} it follows that 
$dens(G[W_j]) \geq dens(G[W'_k])$.
Hence, 
\[
dens(G[W_k]) \geq \frac{1}{2} dens(G[W'_k]).
\]
Now, we claim that $dens(G[W'_k]) \geq dens(G[W^o_k])$.
Assume that this is not the case, and that $dens(G[W'_k]) < G[W^o_k]$.
Notice that at least one of 
$G[W^o_1]$, $G[W^o_2]$, $\dots$, $G[W^o_k]$
does not belong to the set $\{ G[W_1], G[W_2], \dots, G[W_{k-1}]\}$ of subgraphs.
Since the subgraphs are in non increasing order of density, it follows
that an optimal solution  of $\DENSEK$ over instance
$(G, $ $\{ G[W_1]$, $G[W_2]$, $\dots$, 
$G[W_{k-1}] \})$ 
is a subgraph of $G$ having density at least $dens(G[W^o_p])$, for some $p$ with $1 \leq p \leq k$, and that $dens(G[W^o_p]) \geq 
dens(G[W^o_k]) > dens(G[W'_k])$, 
contradicting the optimality of $G[W'_k]$.
Hence it must hold 
\[
dens(G[W_k]) \geq \frac{1}{2} dens(G[W'_k]) \geq 
\frac{1}{2} dens(G[W^o_k]).
\]

Now,
\[
\sum_{i=1}^k dens(G[W_i]) = \sum_{i=1}^{k-1} dens(G[W_i]) +
dens(G[W_k])  \geq \frac{1}{2} \sum_{i=1}^{k-1} dens(G[W^o_i]) +
\frac{1}{2} dens(G[W^o_k]) \geq
\]
\[
\frac{1}{2} \sum_{i=1}^{k} dens(G[W^o_i])
\]
thus concluding the proof.
\end{proof}

\setcounter{Applemma}{14}

\subsection*{Proof of Lemma \ref{lem:3Case1}}

\begin{Applemma}
\label{appendix:lem:3Case1}
Let $G_P=(V_P,E_P)$ be a graph instance of {\sf 3-Clique Partition}
and let 
$G=(V,E)$ be the corresponding graph instance of \TOPT{}.
Given three cliques $G_P[V_{P,1}]$, $G_P[V_{P,2}]$, $G_P[V_{P,3}]$ such that $V_{P,1}$, $V_{P,2}$, $V_{P,3}$ partition $V_P$,
we can compute in polynomial time 
a set $\mathcal{W}= \{G[V_1], G[V_2], G[V_3] \}$ 
such that
$r(\mathcal{W}) \geq \frac{|V| -3}{2} +  18|V|^3$.
\end{Applemma}
\begin{proof}
By construction the three subgraphs $G_P[V_{P,1}]$, $G_P[V_{P,2}]$, $G_P[V_{P,3}]$ of $G_P$ are disjoint.
Construct three subgraphs $G[V_1]$, $G[V_2]$, $G[V_3]$ of $G$ as follows:
\[
V_i = \{ u_j \in V_i: v_j \in V_{P,i}  \}
\]
It follows that $G[V_1]$, $G[V_2]$, $G[V_3]$ are disjoint and
that $V_1 \uplus V_2 \uplus V_3 = V$.
Hence
\[
r(\mathcal{W}) = dens(\mathcal{W})+ \lambda \sum_{i=1}^3 \sum_{j=i+1}^3 d(G[V_i],G[V_j]) 
\]
where 
\[
dens(\mathcal{W}) = dens(G[V_1])+ dens(G[V_2])+ dens(G[V_3])= 
\frac{|E_1|}{|V_1|} + \frac{|E_2|}{|V_2|} + \frac{|E_3|}{|V_3|}.
\]
Since $G_P[V_{P,i}]$, with $1 \leq i \leq 3$, is a clique and, by construction, $G[V_i]$
is also a clique, it follows that 
\[
\frac{|E_i|}{|V_i|} = \frac{|V_i|(|V_i|-1)}{2|V_i|} = \frac{|V_i|-1}{2}
\]
thus
\[
dens(G[V_1])+ dens(G[V_2])+ dens(G[V_3])= 
\frac{|V_1|-1}{2}+\frac{|V_2|-1}{2} +\frac{|V_3|-1}{2}= \frac{|V|-3}{2}.
\]
For each $i,j \in \{ 1,2,3 \}$, with $i \neq j$,
$G[V_i]$ and $G[V_j]$ are disjoint, hence:
\[
d(G[V_i],G[V_j])=2 
\]
Thus, $r(\mathcal{W}) \geq \frac{|V| -3}{2}+  18|V|^3$.
\end{proof}

\setcounter{Applemma}{15}

\subsection*{Proof of Lemma \ref{lem:3Case2}}

\begin{Applemma}
\label{appenix:lem:3Case2}
Let $G_P=(V_P,E_P)$ be a graph instance of $3$-Clique Partition and let 
$G=(V,E)$ be the corresponding graph instance of \TOPT{}.
Given  a solution $\mathcal{W}= \{G[V_1], G[V_2], G[V_3] \}$ of \TOPT{} on
instance $G$, 
with $r(\mathcal{W}) \geq \frac{|V| -3}{2}+  18|V|^3$,
we can compute in polynomial time three cliques
$G_P[V_{P,1}]$, $G_P[V_{P,2}]$, $G_P[V_{P,3}]$ of $G_P$ such that $V_{P,1}$, $V_{P,2}$, $V_{P,3}$  partition $V_P$.
\end{Applemma}
\begin{proof}
First, we have to prove that $G[V_1]$, $G[V_2]$, $G[V_3]$ are disjoint
and that $V_1 \uplus V_2 \uplus V_3 =V$.
Assume to the contrary
that two subgraphs in $\mathcal{W}$, w.l.o.g. $G[V_1]$ and $G[V_2]$,
share at least one vertex. Then
\[
d(G[V_1],G[V_2])=2 - \frac{|V_1 \cap V_2|^2}{|V_1||V_2|} \leq 2 - \frac{1}{|V|^2}.
\]
Since $dens(\mathcal{W}) \leq \frac{3(|V|-1)}{2}$, as $\frac{|E_i|}{|V_i|} \leq \frac{|V_i|-1}{2} \leq \frac{|V|-1}{2}$, and $\lambda=3|V|^3$,
it follows that
\[
r(\mathcal{W}) \leq 3\frac{|V| -1}{2} + 12 |V|^3 + 3|V|^3(2 - \frac{1}{|V|^2})
< 3 |V| + 18 |V|^3 - 3|V|= 18 |V|^3 
\]
Thus $r(\mathcal{W}) < \frac{|V| -3}{2} +  18|V|^3$, contradicting the hypothesis that 
$r(\mathcal{W}) \geq \frac{|V| -3}{2} +  18|V|^3$. Thus we can assume
that $G[V_1]$, $G[V_2]$, $G[V_3]$ are disjoint.

Now, we show that $V_1 \uplus V_2 \uplus V_3 =V$. Assume that this is not
the case. Let $dens(G[V_i]) = z_i$, with $1 \leq i \leq 3$. 
Since $G[V_1]$, $G[V_2]$, $G[V_3]$ are disjoint, it follows
that $z_1 + z_2 + z_3 < \frac{|V| -3}{2}$, thus 
$r(\mathcal{W}) < \frac{|V| -3}{2} +  18|V|^3$ contradicting the
hypothesis that
$r(\mathcal{W}) \geq \frac{|V| -3}{2} +  18|V|^3$. 
Moreover, notice that,
since $r(\mathcal{W}) \geq \frac{|V| -3}{2} +  18|V|^3$,
$dens(\mathcal{W}) \geq \frac{|V| -3}{2}$, thus each $G[V_i]$, with
$1 \leq i \leq 3$,
is a clique in $G$.

Now, we define $G_P[V_{P,1}]$, $G_P[V_{P,2}]$, $G_P[V_{P,3}]$:
\[
V_{P,i} = \{ v_j : u_j \in V_{i}  \}
\]
By construction of $G$, it follows that $G[V_{P,1}]$, $G[V_{P,2}]$, $G[V_{P,3}]$ are disjoint, $V_{P,1} \uplus V_{P,2} \uplus V_{P,3} = V_P$ and that $G[V_{P,i}]$, 
with $1 \leq i \leq 3$, is a clique.
\end{proof}

\end{document}